\documentclass[11pt]{article}
\usepackage[letterpaper,margin=1.00in]{geometry}
\usepackage{graphicx} 
\usepackage{amsthm}
\usepackage{amsmath}
\newtheorem{theorem}{Theorem}[section]

\newtheorem{lemma}[theorem]{Lemma}

\usepackage{kotex}
\usepackage{amsfonts}
\usepackage{amssymb}
\usepackage{amsmath}
\usepackage{graphicx} 

\usepackage{ifthen}

\usepackage{appendix}
\usepackage{graphicx}
\usepackage{color}
\usepackage{algorithm}
\usepackage[noend]{algpseudocode}
\usepackage{wrapfig}
\usepackage{balance}

\usepackage{framed}
\usepackage[framemethod=tikz]{mdframed}
\usepackage[bottom]{footmisc}
\usepackage{enumitem}

\let\oldtextbf=\textbf
\renewcommand\textbf[1]{{\boldmath\oldtextbf{#1}}}

\definecolor{darkgreen}{rgb}{0,0.5,0}
\usepackage{hyperref}
\usepackage{thmtools}
\usepackage{thm-restate}

\hypersetup{
    unicode=false,          
    colorlinks=true,        
    linkcolor=red,          
    citecolor=darkgreen,        
    filecolor=magenta,      
    urlcolor=cyan           
}
\usepackage[capitalize, nameinlink]{cleveref}

\algnewcommand\algorithmicswitch{\textbf{switch}}
\algnewcommand\algorithmiccase{\textbf{case}}

\algdef{SE}[SWITCH]{Switch}{EndSwitch}[1]{\algorithmicswitch\ #1\ \algorithmicdo}{\algorithmicend\ \algorithmicswitch}%
\algdef{SE}[CASE]{Case}{EndCase}[1]{\algorithmiccase\ #1}{\algorithmicend\ \algorithmiccase}%
\algtext*{EndSwitch}%
\algtext*{EndCase}%

\renewcommand{\paragraph}[1]{\vspace{0.15cm}\noindent {\bf #1}.}

\newcommand{\FullOrShort}{full}

\ifthenelse{\equal{\FullOrShort}{full}}{
    
  \newcommand{\fullOnly}[1]{#1}
  \newcommand{\shortOnly}[1]{}
  }{
    \newcommand{\fullOnly}[1]{}
    \newcommand{\shortOnly}[1]{#1}
  }

\title{An Optimal MPC Algorithm for Subunit-Monge Matrix Multiplication, with Applications to LIS}

\begin{document}
\date{}
 \author{Jaehyun Koo \\ \small MIT \\ \small koosaga@mit.edu}
\maketitle

\begin{abstract} 
We present an $O(1)$-round fully-scalable deterministic massively parallel algorithm for computing the min-plus matrix multiplication of unit-Monge matrices. We use this to derive a $O(\log n)$-round fully-scalable massively parallel algorithm for solving the exact longest increasing subsequence (LIS) problem. For a fully-scalable MPC regime, this result substantially improves the previously known algorithm of $O(\log^4 n)$-round complexity, and matches the best algorithm for computing the $(1+\epsilon)$-approximation of LIS.
\end{abstract}

\maketitle

\section{Introduction}
The \textit{longest increasing subsequence} (LIS) problem is a fundamental problem in computer science. A celebrated $O(n \log n)$ sequential algorithm of Fredman \cite{FREDMAN197529} has been known for decades, and the result is tight in the comparison model \cite{Ramanan1997TightL}. On the other hand, in the MPC literature, our understanding of the LIS problem is somewhat limited. Although we've recently seen a lot of advances in this area, all known algorithms still require either approximations, scalability restrictions, a large number of rounds, or combinations of them.

One possible approach for this problem is to use the concept of unit-Monge matrices introduced by Tiskin \cite{tiskin2013semilocal}. Tiskin's framework has a highly desirable property of decomposability, where the LIS problem is decomposed into $O(n)$ multiplication of the unit-Monge matrix, which can be processed in $O(\log n)$ rounds in the MPC model. For the problem of LIS, Tiskin's framework was used to obtain new algorithms in various computing models \cite{tiskin2013semilocal, kt09, kt10, chs23, dynamiclis2021a}, hinting its possible application for the MPC model as well. We note that Tiskin's framework was also used in new algorithms for the longest common subsequence (LCS) problem \cite{tiskin2013semilocal, RUSSO201230, TISKIN2008570}, another fundamental problem closely connected to the LIS problem. 

In this work, we present an $O(1)$-round fully-scalable MPC algorithm for the multiplication of unit-Monge matrix. This result implies an $O(\log n)$ round algorithm to the exact LIS problem in the MPC model, which substantially improves the previously known algorithm of $O(\log^4 n)$ round complexity \cite{chs23} and matching the best approximation algorithm in the fully-scalable MPC regime \cite{ims17}. We hope this result could illustrate the power of the unit-Monge framework and be a bridge for closing the gap on the LIS problem in the MPC model.

\subsection{The MPC Model}
The \textit{Massively Parallel Computation} (MPC) model was first introduced in \cite{10.5555/1873601.1873677}, and later refined in \cite{10.1145/3125644, 10.1145/2591796.2591805,10.1007/978-3-642-25591-5_39}. Our description of the model closely follows that of \cite{10.1145/2591796.2591805}.

In the MPC model, we have $m$ machines with space $s$, where $n$ is the size of the input, $m = O(n^\delta)$ and $s = \tilde{O}(n^{1 - \delta})$ for a parameter $\delta$ (here $\tilde{O}$ hides the poly-log factor). The parameter $\delta$ is generally assumed to be any constant such that $0 < \delta < 1$. If an MPC algorithm works for any such $\delta$, it is called \textit{fully scalable}. All our results here will be fully-scalable.

In the beginning, the input data is distributed across the machines. Then, the computation proceeds in rounds. In each round, a machine performs local computation on its data (of size $s$) and then sends messages to other machines for the next round. The total amount of communication sent or received by a machine is bounded by $s$, its space. In the next round, each machine treats the received messages as the input for the round. 

In the MPC model, the primary complexity measure is the number of rounds $R$ required to solve a problem, corresponding to a \textit{span} in the parallel algorithm literature.

\subsection{Previous works}
\paragraph{Subunit-Monge Matrix Multiplication}  
The seminal work of Tiskin \cite{tiskin2013semilocal} revealed an intriguing structural fact about the Subunit-Monge matrix, its application to the LIS and LCS problem, and an efficient centralized algorithm. Following this work, Krusche and Tiskin \cite{kt09} designed a parallel algorithm that achieves an $O(1)$ superstep subunit-Monge matrix multiplication in a BSP model introduced by Valiant \cite{bsp}. For a BSP algorithm to be applicable in the standard MPC model, it should achieve a \textit{fully-scalable communication and memory} as defined in \cite{kt10} (one should not confuse this with the fully-scalable MPC algorithm). Formally, a BSP algorithm with $p$ machine should achieve $\tilde{O}(\frac{n}{p})$ communication cost and memory cost. Unfortunately, their algorithm requires $\Omega(\frac{n}{\sqrt p})$ memory and communication for $p$ machine, which makes it not fully scalable and is not applicable for an MPC model.

The following work from Krusche and Tiskin \cite{kt10} achieves an $O(\log n)$ superstep BSP algorithm with fully scalable communication and memory. While this work does translate to an MPC algorithm with $O(\log n)$ rounds, their algorithm has a restriction of $p < n^{1/3}$ as their communication and memory cost is $\tilde{O}(\frac{n}{p} + p^2)$ to be precise. Hence, their algorithm is applicable for an MPC model only in a range of $0 < \delta < \frac{1}{3}$, and also does not achieve a $O(1)$ round complexity.

Recently, Cao, Huang, and Su \cite{chs23} obtained an algorithm in an EREW PRAM model with $O(\log^3 n)$ span. This result also implies an MPC algorithm for $O(\log^3 n)$ round complexity. While this is the best algorithm that satisfies the fully scalable property, we note that the round complexity is much larger than the optimal algorithm.

\paragraph{Longest Increasing Subsequence}
Tiskin's work \cite{tiskin2013semilocal} reveals the connection between the subunit-Monge Matrix Multiplication and the Longest Increasing Subsequence (LIS) problem. They show an equivalence relation between the length of LIS and the multiple of $n$ subunit-Monge matrices of total size $O(n)$. Due to this, in most parallel models, an $S(n)$ round algorithm for subunit-Monge matrix multiplication implies an $O(S(n) \log n)$ round algorithm for the length of LIS using the standard divide-and-conquer method. Most of the recent works concerning the exact computation of LIS in the parallel model use this framework - they start by obtaining an efficient subunit-Monge matrix multiplication algorithm and use it to obtain the length of LIS or the actual certificate with a cost of $O(\log n)$ multiplicative factor on rounds. We note that our algorithm for subunit-Monge matrix multiplication can be applied in the same way, and the result from the paper implies an $O(\log n)$ round for the exact computation of LIS length.

In the approximate regime, the paper from Im, Moseley, Sun \cite{ims17} showed an $O(1)$-round algorithm for computing the $(1+\epsilon)$-approximate LIS. Unfortunately, their approach is not fully scalable: The usage of large DP tables in the computation restricts the algorithm to only work in the range $0 < \delta \le \frac{1}{4}$. To date, the best fully-scalable algorithm for $(1+\epsilon)$-approximate LIS has $O(\log n)$ round complexity \cite{ims17}, which is matched by our exact LIS algorithm.

\paragraph{Longest Common Subsequence} 
Hunt and Szymanski \cite{huntszymanski} introduced a simple reduction from the LCS problem to the LIS problem: For two string $S = \{s_1, s_2, \ldots, s_n\}, T = \{t_1, t_2, \ldots, t_n\}$, list all \textit{matching pairs} of $(i, j)$ with $s_i = t_j$ in lexicographical order. In this sequence, a  subsequence of pairs where $j$ is strictly increasing corresponds to each common subsequence of $S$ and $T$. As long as one can hold all matching pairs across the machines, the efficient solution to the LIS problem implies an efficient solution to the LCS problem. For example, \cite{chs23} implies an $O(\log^4 n)$-span parallel algorithm for the LCS problem, given that one can afford $\tilde{O}(n^2)$ total work and the space. 

For $\tilde{O}(n^2)$ total space, Apostolico, Atallah, Larmore, and McFaddin provided an $O(\log^2 n)$-span CREW PRAM parallel algorithm with $O(n^2)$ space \cite{apostolico}, which also translates to an $O(\log^2 n)$-round MPC algorithm. On the other hand, the problem gets significantly more difficult with more stringent space restrictions. No exact algorithm with $o(n)$-round is known for subquadratic total space. Moreover, no constant-factor approximation algorithm with $o(n)$-round is known for the general MPC model with near-linear total space. We note that the framework of \cite{ims17} also does not apply to the LCS problem as it has no decomposability and monotonicity properties. Recently, there was progress on the MPC algorithms in a model that allows constant approximation and super-linear total space. \cite{mpclcs1, mpclcs2, mpclcs3} 

\begin{table*}
\centering\small
\renewcommand{\arraystretch}{1.1}
\begin{tabular}{|l|c|c|c|}
  \hline
   \bf Reference &Number of Rounds & Scalability & Approximation\\
   \hline
  \bf \cite{kt10} & $O(\log^2 n)$ & $\delta < \frac{1}{3}$ & Exact\\
  \hline
  \bf \cite{ims17} & $O(\log n)$ & Fully-scalable & $(1+\epsilon)$\\
  \hline
  \bf \cite{ims17} & $O(1)$ & $\delta < \frac{1}{4}$ & $(1+\epsilon)$\\
  \hline
  \bf \cite{chs23} & $O(\log^4 n)$ & Fully-scalable & Exact\\
  \hline
\bf \textbf{This paper.} & $O(\log n)$ & Fully-scalable & Exact\\
  \hline
\end{tabular}
\caption{\label{tab:o1} \small
\vspace{-1em}Summary of known results for massively parallel LIS algorithms.
}
\end{table*} 

\subsection{Our contribution}
Our main result is the constant-round algorithm for the unit-Monge multiplication of two unit-Monge matrices.

\begin{restatable}[Main Theorem]{theorem}{main}
\label{thm:main}
   Given two permutation matrices $P_A$, $P_B$ represented by the indices of their nonzero entries, there is a fully-scalable deterministic MPC algorithm for computing the implicit unit-Monge matrix multiplication $P_C = P_A \boxdot P_B$ represented as the indices of their nonzero entries in $O(1)$ rounds.
\end{restatable}

Note that the (sub-)permutation matrices are the equivalent representation of (sub)unit-Monge matrices. The exact connection between these two objects will be made clear in the Preliminaries section.

Following \cref{thm:main}, we show how to generalize this result using the techniques from \cite{tiskin2013semilocal}. The only difference from the \cref{thm:main} is that it allows subunit-Monge matrices.

\begin{restatable}{theorem}{mainExt}
\label{thm:main_ext}
   Given two sub-permutation matrices $P_A$, $P_B$ represented by the indices of their nonzero entries, there is a fully-scalable deterministic MPC algorithm for computing the implicit subunit-Monge matrix multiplication $P_C = P_A \boxdot P_B$ represented as the indices of their nonzero entries in $O(1)$ rounds.
\end{restatable}

The result of \cref{thm:main_ext} has various implications as many problems can be represented as a min-plus multiplication of several subunit-Monge matrices. The first and the most notable result is the following:

\begin{restatable}{theorem}{mainLis}
\label{thm:main_lis}
   Given a sequence of $n$ numbers $A = \{a_1, a_2, \ldots, a_n\}$, there is a fully-scalable MPC algorithm for computing the length of the LIS of $A$ in $O(\log n)$ rounds.
\end{restatable}

Using the observation provided by Hunt and Szymanski \cite{huntszymanski}, we can also derive the following:

\begin{restatable}{corollary}{mainLcs}
\label{cor:main_lcs}
    Given two sequences $A, B$ of $n$ numbers, there is a fully-scalable MPC algorithm for computing the length of the LCS of $A$ and $B$ in $O(\log n)$ rounds, given that we have $m = n^{1 + \delta}$ machines and $s = \tilde{O}(n^{1-\delta})$ space for any $0 < \delta < 1$.
\end{restatable}

In a case where each machine has $O(n)$ memory, the classical dynamic programming on a single machine suffices. Previously, only $O(\log^2 n)$-round algorithms were known in this setting \cite{apostolico}. Note that \cref{cor:main_lcs} do not fit in our definition of the MPC model since the total required space is $\tilde{O}(n^2)$. On the other hand, as we noted earlier, no known algorithms use near-linear total space and have $o(n)$ round complexity, even if we allow constant-factor approximation. 

Finally, the structure we used to compute the length of LIS can be used to solve the generalized version of \textit{semi-local} LIS and LCS problem \cite{tiskin2013semilocal, chs23}. In the semi-local LCS problem, given two strings $S$ and $T$, the algorithm should compute the LCS between $S$ and every subsegment of $T$, where a subsegment of $T$ is an array of form $[T_i, T_{i+1}, \ldots, T_j]$ for some $1 \le i \le j \le |T|$. In the semi-local LIS problem, the algorithm should compute the LIS for all subsegments of $T$. 

\begin{restatable}{corollary}{mainSemilocalLis}
\label{cor:main_semilocal_lis}
   Given a sequence of $n$ numbers $A = \{a_1, a_2, \ldots, a_n\}$, there is a fully-scalable MPC algorithm for computing the semi-local LIS of $A$ in $O(\log n)$ rounds.
\end{restatable}

\begin{restatable}{corollary}{mainSemilocalLcs}
\label{cor:main_semilocal_lcs}
    Given two sequences $A, B$ of $n$ numbers, there is a fully-scalable MPC algorithm for computing the semi-local LCS of $A$ and $B$ in $O(\log n)$ rounds, given that we have $m = n^{1 + \delta}$ machines and $s = \tilde{O}(n^{1-\delta})$ space for any $0 < \delta < 1$.
\end{restatable}

\subsection{Technical Overview}
\textbf{Warmup: An $O(\log n)$-round algorithm for unit-Monge matrix multiplication.}
In the Section 3.1 of \cite{chs23}, they identified the following \textit{core problem}: Given two permutation $P_1, P_2$ of size $n$, we define $f(i)$ to be the smallest integer $k$ such that $\sum_{j < i}\mathbb{I}[P_1[j] < k] > \sum_{j \geq i} \mathbb{I}[P_2[j] \geq k]$. The problem is to find all $f(i)$ for $0 \le i \le n$, and \cite{chs23} provides an $O(\log^2 n)$-span divide-and-conquer algorithm to do this, which implies an $O(\log^3 n)$-span algorithm for subunit-Monge matrix multiplication. Indeed, we provide an $O(1)$-round MPC algorithm to compute $f(i)$, which implies an $O(\log n)$-round algorithm for subunit-Monge matrix multiplication.

First, let's start from an $O(\log n)$-round algorithm to compute the $f(i)$. Using binary search, we can reduce the problem into the computation of $\sum_{j < i}\mathbb{I}[P_1[j] < f(i)] - \sum_{j \geq i} \mathbb{I}[P_2[j] \geq f(i)]$ in $O(1)$ rounds.

To compute this, we construct a binary tree over $P_1$ and $P_2$, where each node contains a sorted list of all leaves in the subtree. Intuitively, this can be considered as a memoization of \textit{merge sort} procedure on each sequence, where each node takes a sorted sequence of two child nodes and merges them as in the merge sort algorithm. This tree can be constructed in $O(\log n)$ rounds as sorting takes $O(1)$ rounds (\cref{lem:tools_sort}). As our problem can be described as counting a value in a range that is smaller than the threshold, we can compute the desired quantity by identifying $O(\log n)$ relevant nodes and doing a rank search in parallel (\cref{lem:tools_search}). Hence, each binary search takes $O(1)$ rounds.

To optimize this procedure, we need to optimize the tree structure construction and the binary search. First, we \textit{flatten} the tree in a way such that each non-leaf node has $H = n^{(1 - \delta) / 10}$ child instead of two. Since the tree now has a depth of $O(1)$, the construction of the tree takes $O(1)$ rounds.

Next, rather than eliminating binary search, we adapt the search to take advantage of this new tree structure. Let's build a tree over the inverse permutation $P_1^{-1}, P_2^{-1}$ instead of the original permutation: This effectively creates a \textit{transposed} version of the tree where for each element $P_1[i]$, the index is $P_1[i]$, and the corresponding value is $i$ (and respectively for $P_2$). We will maintain two pointers, one for the tree over $P_1^{-1}$ and another for $P_2^{-1}$, and descend to one of $H$ child by counting the value less (or greater than $j$) and descending toward the leftmost child that satisfies the predicate.

This procedure ultimately yields $O(1)$ rounds for each query, but there is a catch - in each round, we are doing $O(n^{1+(1-\delta)/10})$ binary searches, which exceeds our memory capacity. Hence, we have to limit our query locations to indices that are a multiple of $G = n^{1-\delta}$. 

Given that we found an $O(1)$-round algorithm on finding $f(iG)$ for all $0 \le i \le \frac{n}{G}$, we need to use this information to compute all $f(i)$. The function $f(i)$ is monotonically decreasing, and what we've found could be visualized as an intercept of axis-parallel polyline formed by $(i, f(i))$ and $\frac{n}{G}+1$ axis-parallel lines $x = 0, G, 2G, \ldots, n$. By taking a transpose of the permutation $P_1, P_2$, we can also compute the same intercepts for lines $y = 0, G, 2G, \ldots, n$. These lines form $(\frac{n}{G})^2$ boxes of size $G \times G$ - As our polyline is monotone, there are at most $\frac{n}{G}$ boxes that our polylines intersect. Inside these boxes, we can compute the inner route of the polylines with $O(G)$-sized information, which is enough for each machine to compute in a single round.

\textbf{Shaving off the $O(\log n)$-factor.} The above algorithm essentially combines two subproblems in $O(1)$ rounds, which is the reason for its $O(\log n)$ round complexity. Our strategy is simple - we will combine $H = n^{(1-\delta)/10}$ subproblems in $O(1)$ rounds, which will solve the entire problem in $O(1)$ round complexity. This is our main technical obstacle.

The actual meaning of the \textit{core problem} in \cite{chs23} is that it computes the demarcation line between the cells that takes its optimum from the first subproblem and the second subproblem: A cell $(x, y)$ in the result matrix will take its value from the first subproblem if $y < f(x)$, and the other if $y \ge f(x)$. For $poly(n)$ subproblems, we prove a similar argument: There exist $q-1$ monotone demarcation lines, where the $i$-th line divides the matrix into the parts where it takes its value from the first $i$ subproblems or the last $q-i$ subproblems. In \cref{sec41}, we describe several definitions and useful lemmas to use this monotonicity property.

Intuitively, these demarcation lines look easy to compute since we can compare the $i$-th subproblem and the $i+1$-th subproblem. However, this approach is incorrect. Consider two cells $(x, y)$ and $(x, y+1)$, where the first cell takes its optimum from the $2$-nd subproblem and the second cell takes its optimum from the $6$-th subproblem. While it seems that for all $2 \le i \le 5$, the $i$-th subproblem should be better for $(x, y)$ compared to $i+1$-th and vice versa, it is not necessarily true, and indeed, all except one can be false. To solve this issue, we instead take the demarcating polyline for all pairs of subproblems using a procedure similar to the Warmup algorithm. Then, for each axis-parallel grid line where its $x$-coordinate (or $y$-coordinate) is a multiple of $G$, we compute the intercept solely based on the abovementioned information. We describe these procedures in detail on \cref{sec42}.

Now, the final piece is to compute the inner route of the polylines using these intercepts. In the Warmup section, a straightforward way based on the definition yields a total of $O(n)$ memory usage. However, since the number of polylines is large, directly applying this algorithm requires $O(n^{1+(1-\delta)/10})$ total memory. In \cref{sec43}, we use several clever observations to reduce the total memory requirements from $O(n^{1+(1-\delta)/10})$ to $O(n)$. 

\section{Preliminaries}

\subsection{Notations}
We follow the notations used in the prior work \cite{chs23}. For integers $i$ and $j$, we use $[i : j]$ to denote set $\{i, i + 1, \ldots, j\}$ and $[i : j)$ to denote set $\{i, i + 1, \ldots, j - 1\}$. The set $[1 : i]$ is represented short as $[i]$. We denote the half-integer set $\{i + \frac{1}{2}, i + \frac{3}{2}, \ldots, j - \frac{1}{2}\}$ by $\langle i : j\rangle$. For brevity, we may denote $[1:i]$ as $[i]$, and $\langle 0:i\rangle$ as $\langle i \rangle$. For an integer variable $i$, the corresponding half-integer variables are denoted by $\hat{i}$, and they have a value of $i + \frac{1}{2}$. Some matrices in this paper are indexed using half-integers. For a set $X$ and $Y$, we denote their Cartesian product as $X \times Y$.

Given a matrix $M$ of size $m \times n$ which is indexed by half-integers of range $\langle 0 : i\rangle \times \langle0 : j\rangle$, its distribution matrix $M^\Sigma$ is defined as

\begin{align*}
M^\Sigma(i, j) = \sum_{(\hat{i}, \hat{j}) \in \langle i : m \rangle \times \langle 0:j\rangle} M(\hat{i}, \hat{j})
\end{align*}

for all $i \in [0:m]$ and $j \in [0:n]$.

A $n \times n$ matrix $P$ is a sub-permutation matrix if and only if 
\begin{itemize}
    \item Each element in $P$ is either $0$ or $1$.
    \item There is at most one element equal to $1$ in each row and each column of $P$.
\end{itemize}

Similarly, a $n \times n$ matrix $P$ is a permutation matrix if and only if 
\begin{itemize}
    \item Each element in $P$ is either $0$ or $1$.
    \item There is exactly one element equal to $1$ in each row and each column of $P$.
\end{itemize}

In our algorithm, we represent them as an array of size $n$, where each index $i$ represents the index of the column of the nonzero element in row $\hat{i}$ when the row $\hat{i}$ is nonzero.

The distribution matrix $P^\Sigma$ of a (sub-)permutation matrix $P$ is called a (sub)unit-Monge matrix. Given two (sub-)permutation matrices $P_A, P_B$, the implicit (sub)unit-Monge matrix multiplication problem is to find a (sub-)permutation matrix $P_C$ that satisfies $P_C^\Sigma(i, k) = \min_j(P_A^\Sigma(i, j) + P_B^\Sigma(j, k))$, or a $(\text{min}, +)$-product of (sub)unit-Monge matrix $P_A^\Sigma$ and $P_B^\Sigma$. Tiskin \cite{tiskin2013semilocal} proved the following:

\begin{lemma}
    For any permutation matrix $P_A, P_B$, there exists a permutation matrix $P_C$ such that $P_C^\Sigma(i, k) = \min_j(P_A^\Sigma(i, j) + P_B^\Sigma(j, k))$.
\end{lemma}

\begin{lemma}
    For any sub-permutation matrix $P_A, P_B$, there exists a sub-permutation matrix $P_C$ such that $P_C^\Sigma(i, k) = \min_j(P_A^\Sigma(i, j) + P_B^\Sigma(j, k))$
\end{lemma}

Since a distribution matrix uniquely defines a sub-permutation matrix, $P_C$ is unique. We use $\boxdot$ to represent the implicit subunit-Monge matrix multiplication. In other words, we denote such $P_C$ as $P_C = P_A \boxdot P_B$.

\subsection{Basic tools}
In this section, we list some basic tools which we will repeatedly invoke in the following sections.

\begin{lemma}[Inverse Permutation]
\label{lem:tools_inv}
    Given a permutation $p : [n] \rightarrow [n]$, we can compute the inverse permutation $p^{-1}$ in $O(1)$ rounds. 
\end{lemma}
\begin{proof}
    Each machine does the following: For all permutation elements it contains, send the pair $(i, p_i)$ to the machine responsible for the $p_i$-th element. 
\end{proof}

We use the results from \cite{10.1007/978-3-642-25591-5_39} in the following lemmas.
\begin{lemma}[Prefix Sums]
\label{lem:tools_prefix}
    Given $n$ numbers, there is a deterministic algorithm to compute a prefix sum of the numbers. \cite{10.1007/978-3-642-25591-5_39}
\end{lemma}

\begin{lemma}[Sorting]
\label{lem:tools_sort}
    Given $n$ numbers, there is a deterministic algorithm to sort them in $O(1)$ rounds. \cite{10.1007/978-3-642-25591-5_39}
\end{lemma}

In this paper, we define the \textit{offline rank searching} on array as the following: Given an array of numbers $A = \{a_1, a_2, \ldots, a_n\}$ and queries $Q = \{q_1, q_2, \ldots, q_k\}$, we should compute the length-$k$ array $R = \{r_1, r_2, \ldots, r_k\}$ where $r_i$ is equal to the number of elements in $A$ which is smaller than $q_i$. Using the lemmas above, we can solve the offline rank-searching problem in $O(1)$ rounds.

\begin{lemma}[Offline Rank Searching]
\label{lem:tools_search}
    Given $n$ numbers and $k$ query points, there is a deterministic algorithm to solve the offline rank searching problem in $O(1)$ rounds.
\end{lemma}
\begin{proof}
Duplicate all numbers and query points and sort them in a nondecreasing order of values. If two elements have the same value, ties are broken so that the query point precedes the number. By \cref{lem:tools_sort} this takes $O(1)$ rounds. Now, when we consider each number as a point of value $1$ and queries as a point of value $0$, the answer $r_i$ corresponds to a prefix sum of such values in a given query point. Computing a prefix sum can be done in $O(1)$ rounds by \cref{lem:tools_prefix}.
\end{proof}

\section{Proof of the Main Theorem}
In this section, we prove the following main theorem:

\main*

\subsection{Preliminaries}\label{sec41}
Let $G = n^{(1 - \delta)}$ and $H = n^{(1-\delta) / 10}$. We will split $P_A$ into $H$ matrices of size $n \times \frac{n}{H}$ in parallel to columns, and $P_B$ into $H = n^{(1-\delta) / 10}$ matrices of size $\frac{n}{H} \times n$ in parallel to rows. We assume $n \ge G = H^{10}$ as we can do all computations locally otherwise. We also assume $n$ to be the multiple of $H$, but this condition is unnecessary and only serves to simplify the explanation. We will denote the splitted matrices of $P_A$ as $P_{A, 1}, P_{A, 2}, \ldots, P_{A, H}$ from left-to-right order, and $P_B$ as $P_{B, 1}, P_{B, 2},\ldots, P_{B, H}$ from top-to-bottom order. Let $P_{C, i} = P_{A, i} \boxdot P_{B, i}$ for all $i \in [H]$. For all $i \in [H]$ $P_{C, i}$ is a $n \times n$ matrix.

We need to remove the empty rows of $P_{A, i}$ and columns of $P_{B, i}$ to make them compact enough for recursion, which we will denote as $P^\prime_{A, i}$ and $P^\prime_{B, i}$. For this, we sort the nonzero row index of each submatrix $P_{A, i}$ and relabel the row accordingly by the rank of their index in the sorted indices. Both requires $O(1)$ round by \cref{lem:tools_inv} and \cref{lem:tools_sort}. After the procedure, we can create a mapping $M_A : \langle n \rangle \rightarrow [H] \times \langle \frac{n}{H} \rangle$ that takes the row index $\hat{i}$ of $P_A$ and returns the pair $M_A(\hat{i}) = (M_A(\hat{i})_{sub}, M_A(\hat{i})_{idx})$, which indicates that the $\hat{i}$-th row is in $M_A(\hat{i})_{sub}$-th subproblem and is indexed as $M_A(\hat{i})_{idx}$-th row there. We also define the inverse of such mapping $M_A^{-1} = [H] \times \langle \frac{n}{H} \rangle \rightarrow \langle n \rangle$. We proceed identically for the columns of $P_B$, creating a mapping $M_B : \langle n \rangle \rightarrow [H] \times \langle \frac{n}{H} \rangle$ and the inverse $M_B^{-1} = [H] \times \langle \frac{n}{H} \rangle \rightarrow \langle n \rangle$.

As each $P^\prime_{A, i}$ and $P^\prime_{B, i}$ are of size $\frac{n}{H} \times \frac{n}{H}$, we can recursively compute $P^\prime_{C, i} = P^\prime_{A, i} \boxdot P^\prime_{B, i}$ for each $i \in [H]$. By the method we constructed, we can derive a relation between $P_{C, i}$ and $P^\prime_{C, i}$ as $P_{C, i}(\hat{x}, \hat{y}) = P^\prime_{C, i}(M_A(\hat{x})_{idx}, M_B(\hat{y})_{idx})$ for $\hat{x}, \hat{y}$ with $M_A(\hat{x})_{sub} = M_B(\hat{y})_{sub} = i$, and $0$ otherwise.

Given the result of subproblem $P_{C, i}$ and the mapping $M_A, M_A^{-1}, M_B, M_B^{-1}$, we need to combine them to the whole subproblem. Ideally, we want a formula representing $P_C$ only using the $P_{C, i}$. We show that this is possible.

\begin{lemma}\label{lem:idk}

For all $0 \le i, k \le n, 1 \le q \le H$, it holds that 
\begin{align*}\min_{0 \le j \le \frac{n}{H}}(P_A^\Sigma(i, (q-1)\frac{n}{H} + j) + P_B^\Sigma((q-1)\frac{n}{H} + j, k)) \\
= \sum_{1 \le x \le q - 1} P_{C, x}^\Sigma(i, n) + P^\Sigma_{C, q}(i, k) + \sum_{q+1 \le x \le H} P_{C, x}^\Sigma(0, k)
\end{align*}
\end{lemma}
\begin{proof}
From the definition, we obtain 
\begin{multline*}\min_{0 \le j \le \frac{n}{H}}(P_A^\Sigma(i, (q-1)\frac{n}{H} + j) + P_B^\Sigma((q-1)\frac{n}{H} + j, k)) \\= 
\sum_{1 \le x \le q - 1} P_{A, x}^\Sigma(i, \frac{n}{H}) + \min_{0 \le j \le \frac{n}{H}}(P_{A, q}^\Sigma(i, j) + P_{B, q}^\Sigma(j,k)) \\+ \sum_{q+1 \le x \le H } P_{B, x}^\Sigma(0, k)
\end{multline*}

We show that $P_{A, x}^\Sigma(i, \frac{n}{H}) = P_{C, x}^\Sigma(i, n)$ for all $x \in [H]$. If we expand the term by definition, we have
\begin{align*}
P^\Sigma_{C, x}(i, n)= \min_{0 \le j \le n/H}(P^\Sigma_{A, x}(i, j) + P^\Sigma_{B, x}(j, n))
\end{align*}

Since the rows of $P_{B, x}$ sums to $1$, we have $P^\Sigma_{B, x}(j, n) = \frac{n}{H} - j$. On the other hand, we have $P^\Sigma_{A, x}(i, j+1) - P^\Sigma_{A, x}(i, j) \le 1$. Thus the RHS is monotonically decreasing over $j$, and the minimum is attained in $j = \frac{n}{H}$. The part $P_{B, x}^\Sigma(0, k) = P_{C, x}^\Sigma(0, k)$ can be proved in an identical way. $P_{C, q}(i, k) = \min_{0 \le j \le n/H}(P_{A, q}(i, j) + P_{B, q}(j,k))$ follows from the definition $P_{C, q} = P_{A, q} \boxdot P_{B, q}$.
\end{proof}

\begin{lemma}
\label{lem:pointwisemin}
For $1 \le q \le H$ and $0 \le i, j \le n$, let $F_{q}(i, j)$ be

\begin{align*}
    F_{q}(i, j) = \sum_{1 \le x \le q-1} P_{C, x}^\Sigma(i, n) + P^\Sigma_{C, q}(i, j) + \sum_{q+1 \le x \le H} P_{C, x}^\Sigma(0, j)
\end{align*}

For all $0 \le i, j \le n$, it holds that $P_C^\Sigma(i, j) = \min_{1 \le q \le H} F_q(i, j)$.
\end{lemma}
\begin{proof}
    Follows from \cref{lem:idk}
\end{proof}

\cref{lem:pointwisemin} shows that we can simply take a pointwise minimum of $F_q(i, j)$ to obtain a distribution matrix. Using such an approach directly requires a lot of space for each machine. Instead, let's define $opt(i, j)$ be the minimum integer in $[q]$ such that $F_{opt(i, j)}(i, j)  = \min_{1 \le q \le H} F_q(i, j)$. We will show that $opt(i, j)$ is monotonic on $i$ and $j$ and derive several nice properties that simplify its computation and representation. Then, we will show that we don't need the value $F_q(i, j)$, and the monotonic table $opt(i, j)$ is enough to produce the resulting permutation matrix $P_C$.

For $1 \le q < r \le H$ and $0 \le i, j \le n$, let $\delta_{q, r}(i, j) = F_{q}(i, j) - F_{r}(i, j)$, which in a more explicit form:
\begin{align*}
    \delta_{q, r}(i, j) = P^\Sigma_{C, q}(i, j) + \sum_{q+1 \le x \le r} P_{C, x}^\Sigma(0, j) 
    - \sum_{q \le x \le r-1} P_{C, x}^\Sigma(i, n) - P^\Sigma_{C, r}(i, j)
\end{align*}

We will prove two lemmas related to the monotonicity of $\delta_{q, r}$.

\begin{lemma}
\label{lem:inc1}
    For all $1 \le q < r \le H$ and $0 \le i \le n, 0 \le j \le n - 1$, we have $0 \le \delta_{q, r}(i, j+1) - \delta_{q, r}(i, j) \le 1$.
\end{lemma}

\begin{proof}
    We have 
    \begin{align*}
        \delta_{q, r}(i, j+1) - \delta_{q, r}(i, j) = (P^\Sigma_{C, q}(i, j+1) - P^\Sigma_{C, q}(i, j)) \\+ \sum_{q+1 \le x \le r} (P_{C, x}^\Sigma(0, j+1) - P_{C, x}^\Sigma(0, j))  \\- (P^\Sigma_{C, r}(i, j+1) -P^\Sigma_{C, r}(i, j))
    \end{align*}
    The value of RHS only depends on the point in the column $\hat{j}$. We divide the cases by $M_B(\hat{j})_{sub}$:
    \begin{itemize}
        \item if $M_B(\hat{j})_{sub} < q$, the RHS is $0$.
        \item if $M_B(\hat{j})_{sub} = q$, the RHS is $1$ if the corresponding point has a row index of at least $i$ and $0$ otherwise.
        \item if $q < M_B(\hat{j})_{sub} < r$, the RHS is $1$.
        \item if $M_B(\hat{j})_{sub} = r$, the RHS is $1$ if the corresponding point has a row index of at most $i$ and $0$ otherwise.
        \item if $r < M_B(\hat{j})_{sub}$, the RHS is $0$.
    \end{itemize}
    Hence, for all cases, it is either $0$ or $1$.
\end{proof}

\begin{lemma}
\label{lem:inc2}
    For all $1 \le q < r \le H$ and $0 \le i \le n - 1, 0 \le j \le n$, we have $0 \le \delta_{q, r}(i+1, j) - \delta_{q, r}(i, j) \le 1$.
\end{lemma}

\begin{proof}
    We have 
    \begin{align*}
        \delta_{q, r}(i+1, j) - \delta_{q, r}(i, j) = -(P^\Sigma_{C, q}(i, j) - P^\Sigma_{C, q}(i+1, j)) \\ +\sum_{q \le x \le r - 1} (P_{C, x}^\Sigma(i, n) - P_{C, x}^\Sigma(i+1, n))  \\+ (P^\Sigma_{C, r}(i, j) -P^\Sigma_{C, r}(i+1, j))
    \end{align*}
    The value of RHS only depends on the point in the row $\hat{i}$. We divide the cases by $M_A(\hat{i})_{sub}$:
    \begin{itemize}
        \item if $M_A(\hat{i})_{sub} < q$, the RHS is $0$.
        \item if $M_A(\hat{i})_{sub} = q$, the RHS is $1$ if the corresponding point has a column index of at least $j$ and $0$ otherwise.
        \item if $q < M_A(\hat{i})_{sub} < r$, the RHS is $1$.
        \item if $M_A(\hat{i})_{sub} = r$, the RHS is $1$ if the corresponding point has a column index of at most $j$ and $0$ otherwise.
        \item if $r < M_A(\hat{i})_{sub}$, the RHS is $0$.
    \end{itemize}
    Hence, for all cases, it is either $0$ or $1$.
\end{proof}

From \cref{lem:inc1} and \cref{lem:inc2}, we can observe that $\delta_{q, r}(i, j)$ is an increasing function over $i$ and $j$. Hence, two areas with $\delta_{q, r}(i, j) < 0$ and $\delta_{q, r}(i, j) \geq 0$ are divided by a monotone chain that starts from $(n, 0)$ and ends at $(0, n)$ while moving in an upper-right direction. 

We can derive the following from \cref{lem:inc1} and \cref{lem:inc2}.
\begin{lemma}
\label{lem:opt1}
    For all $0 \le i \le n - 1, 0 \le j \le n$, we have $opt(i, j) \le opt(i + 1, j)$.
\end{lemma}
\begin{proof}
Suppose $opt(i, j) > opt(i, j + 1)$. Let $k_1 = opt(i, j + 1), k_2 = opt(i, j)$, we have\\ $\delta_{k_1, k_2}(i, j + 1) \le 0$ and $\delta_{k_1, k_2}(i, j) > 0$, which contradicts \cref{lem:inc1}.
\end{proof}
\begin{lemma}
\label{lem:opt2}
    For all $0 \le i \le n, 0 \le j \le n - 1$, we have $opt(i, j) \le opt(i, j+1)$.
\end{lemma}
\begin{proof}
Apply the proof of \cref{lem:opt1}, this time using \cref{lem:inc2} instead.
\end{proof}

To give the complete characterization of points with $P_C(\hat{i}, \hat{j}) = 1$ by only using the table of $opt$, we derive several lemmas.

\begin{lemma}
\label{lem:char1}
    For all $0 \le i, j \le n - 1$, if $opt(i, j) \neq opt(i, j + 1)$ and $P_{C, q}(\hat{i}, \hat{j}) = 0$ for all $q$, it holds that $P_C(\hat{i}, \hat{j}) = 0$.
\end{lemma}
\begin{proof}
Let $a= opt(i, j), b = opt(i, j+1), c = opt(i+1, j), d = opt(i+1, j+1)$. By \cref{lem:opt1} and \cref{lem:opt2}, we have $a < b, a \le c, c \le d, b \le d$.

We first prove $F_b(i, j+1) = F_d(i, j+1)$. If $b = d$, this is evident. Suppose that $b \neq d$. Since $b < d$, we have that $\delta_{b, d}(i, j+1) \le 0$ and $\delta_{b, d}(i+1, j+1) > 0$. By \cref{lem:inc2}, the only possible choice of value is $\delta_{b, d}(i, j+1) = 0, \delta_{b, d}(i+1, j+1) = 1$. This means that $F_b(i, j+1) = F_d(i, j+1)$. With the same procedure, this time using  \cref{lem:inc1}, we can also prove $F_c(i+1, j) = F_d(i+1, j)$ as well.

Now we prove $F_a(i, j) = F_d(i, j)$. Since $a < b$, we have that $\delta_{a, b}(i, j) \le 0$ and $\delta_{a, b}(i, j+1) > 0$. By \cref{lem:inc1}, the only possible choice of value is $\delta_{a, b}(i, j) = 0, \delta_{a, b}(i, j+1) = 1$. Since $F_b(i, j+1) = F_d(i, j+1)$, we have $\delta_{a, d}(i, j+1) = 1$. By \cref{lem:inc1}, $\delta_{a, d}(i, j) \geq 0$, and we can't have $\delta_{a, d}(i, j) > 0$ as that will make $a$ not a minimum. Hence we have $\delta_{a, b}(i, j) = 0$ and $F_a(i, j) = F_d(i, j)$. Summing this up, we have:

\begin{align*}
P_C^\Sigma(i, j) = F_d(i, j)\\
P_C^\Sigma(i, j+1) = F_d(i, j+1)\\
P_C^\Sigma(i+1, j) = F_d(i+1, j)\\
P_C^\Sigma(i+1, j+1) = F_d(i+1, j+1)
\end{align*}

By \cref{lem:pointwisemin}, we have $P_C(\hat{i}, \hat{j}) = P_{C, d}(\hat{i}, \hat{j})$, and since the RHS is $0$, the LHS is $0$ as well.
\end{proof}

\begin{lemma}
\label{lem:char2}
    For all $0 \le i, j \le n - 1$, if $opt(i, j) \neq opt(i+1, j)$ and $P_{C, q}(\hat{i}, \hat{j}) = 0$ for all $q$, it holds that $P_C(\hat{i}, \hat{j}) = 0$.
\end{lemma}
\begin{proof}
Proceed as in \cref{lem:pointwisemin}.
\end{proof}

\begin{lemma}
\label{lem:char3}
    For all $0 \le i, j \le n - 1$, if $opt(i, j) = opt(i+1, j), opt(i, j) = opt(i, j+1)$ and $opt(i, j) \neq opt(i+1, j+1)$, it holds that $P_C(\hat{i}, \hat{j}) = 1$.
\end{lemma}
\begin{proof}
Let $a = opt(i, j) = opt(i+1, j) = opt(i, j+1)$ and $b = opt(i+1, j+1)$. By \cref{lem:inc1}, we have $a < b$. By the statement we have that $\delta_{a, b}(i+1, j) \le 0$ and $\delta_{a, b}(i+1, j+1) > 0$. By \cref{lem:inc1}, the only possible choice of value is $\delta_{a, b}(i+1, j) = 0, \delta_{a, b}(i+1, j+1) = 1$. Hence, we have:

\begin{align*}
P_C^\Sigma(i, j) = F_a(i, j)\\
P_C^\Sigma(i, j+1) = F_a(i, j+1)\\
P_C^\Sigma(i+1, j) = F_a(i+1, j) \\
P_C^\Sigma(i+1, j+1) = F_a(i+1, j+1) + 1
\end{align*}

By \cref{lem:pointwisemin}, we have $P_C(\hat{i}, \hat{j}) = P_{C, a}(\hat{i}, \hat{j}) + 1$. Since $P_{C, a}(\hat{i}, \hat{j})$ is either $0$ or $1$, this leaves the only option of $P_C(\hat{i}, \hat{j}) = 1$.

\end{proof}

\begin{lemma}
\label{lem:char4}
    For all $0 \le i, j \le n - 1$, if $opt(i, j) = opt(i+1, j)$, $opt(i, j) = opt(i, j+1)$, $opt(i, j) = opt(i+1, j+1)$, then $P_C(\hat{i}, \hat{j}) = 1$ holds if and only if $P_{C, opt(i, j)}(\hat{i}, \hat{j}) = 1$.
\end{lemma}
\begin{proof}
Let $a = opt(i, j)$, we have:
\begin{align*}
P_C^\Sigma(i, j) = F_a(i, j)\\
P_C^\Sigma(i, j+1) = F_a(i, j+1)\\
P_C^\Sigma(i+1, j) = F_a(i+1, j) \\
P_C^\Sigma(i+1, j+1) = F_a(i+1, j+1)
\end{align*}

Hence by \cref{lem:pointwisemin}, we have $P_C(\hat{i}, \hat{j}) = P_{C, a}(\hat{i}, \hat{j})$.
\end{proof}

\cref{lem:char1}, \ref{lem:char2}, \ref{lem:char3}, \ref{lem:char4} gives a sufficient characterization for all points with $P_C(\hat{i}, \hat{j}) = 1$. Except the points added in the procedure of \cref{lem:char3}, for all points in $P_C(\hat{i}, \hat{j}) = 1$ there exists a point in some $q$ where $P_{C, q}(\hat{i}, \hat{j}) = 1$. Hence, the algorithm can start by simply taking all points characterized by \cref{lem:char3} and then fill the remaining points from the permutation matrix $\sum_{q = 1}^{H} P(C, q)$.

\subsection{Computing the $opt(i, j)$ for selected lines}
\label{sec42}
For the first $n^{\delta} + 1$ machines indexed as $j = 0, 1, \ldots, n^{\delta}$, we will devise an algorithm that computes $opt(*, jG)$ in each machines. As the tables can be large, they won't be stored explicitly but rather as a set of $H$ intervals where the $k$-th interval denotes the set of $i$ with $opt(i, jG) = k$. 

To this end, for each $1 \le q < r \le H$, we will try to find the first $i$ where $\delta_{q, r}(i, jG) > 0$. Let $cmp(jG, q, r)$ be such $i$ ($n+1$ if there is no such $i$). If we compute all such $O(H^2)$ values, we can determine the $opt(i, jG)$ for all possible $j$ alongside the possible location where this minimum might change. By directly trying all possibilities, we can compute the desired intervals, and since $H^2 \le G$, this does not violate the space condition.

Let $p$ be a permutation on $\langle n\rangle$ created by taking a union of all points in each $P_{C, i}$. In other words, for all $i, \hat{x}, \hat{y}$, $P_{C, i}(\hat{x}, \hat{y}) = 1$, then $p(\hat{x}) = \hat{y}$. Additionally, to record the origin of each point, we say $p(\hat{x})$ is of \textbf{color} $i$ if $P_{C, i}(\hat{x}, \hat{y}) = 1$.

We will construct a $H$-ary tree structure $T$, where each node at the $i$-th level represents $H^i$ consecutive element in the permutation $p_i$. For example, each nodes at $0$-th level represents one element, such as $\{p(\hat{0})\}, \{p(\hat{1})\}$ and so on. In the $1$-th level, each nodes represent $\{p(\hat{0}), p(\hat{1}), \ldots, p(\hat{H}-1)\}$, $\{p(\hat{H}), p(\hat{H}+1), \ldots, p(\hat{2H}-1)\}$ and so on. This tree has a height of $\frac{1}{1 - \frac{\delta}{10}} = O(1)$. We denote $h(T)$ to be the height of the tree, hence $h(T) = \frac{1}{1 - \frac{\delta}{10}}$. If two nodes in an adjacent level share an element, the node in the lower level is the \textit{child} of the higher level. We say some child is left or right to the other child if it points to an interval with lower or higher indices. 

For each node $v$ of $T$, we associate it to $H$ arrays where the $i$-th array $arr(v, i)$ contains all elements of node $v$ with color $i$, sorted in the increasing order of $p_i$. $T$ can be constructed in $O(1)$ round. First, we create the leaf nodes. Then, the construction proceeds in phases. We do the following in the $i$-th phase ($1 \le i \le h(T)$):

\begin{itemize}
    \item Each node at $i-1$-th level clones itself. 
    \item Each node $v$ at $i$-th level collects all cloned information from the $i-1$-th level. It identifies each $arr(w, i)$ for all children $w$ of $v$ and sorts them all together to create $arr(v, i)$, which can be done in $O(1)$ round by \cref{lem:tools_sort}.
\end{itemize}

For the space usage, we need to analyze two parts: The size of the node itself and the total length of arrays associated with each node. For the second part, note that the sum of size of $arr(v, i)$ is $n$ for all $0 \le i \le h(T)$. Hence, the size of the arrays fits the near-linear space usage. For the first part, the total space usage is super-linear as each node contains $H$ pointers toward the arrays, and we have $O(n)$ nodes. To this end, we implicitly represent the lower $10$ levels without constructing a tree since each $H^{10} = G$ element consisting of the level-$10$ nodes of the tree can fit in a single machine. We do not create an actual tree for these single machines, and we process all the potential queries solely based on the information over the $H^{10}$ elements it contains.

After constructing the tree $T$, we will search the value $cmp(jG, q, r)$ by descending the tree $T$ in parallel. Specifically, the algorithm operates in $h(T)$ phases. Let $V(jG, q, r)$ be the current node of $T$, where we know that $cmp(jG, q, r)$ is in the interval that node represents. At each phase, $V(jG, q, r)$ descends to one of the children while satisfying one invariant: Let $\langle l:r \rangle$ be the range of index the node is representing, then $V(jG, q, r)$ is the rightmost node with $\delta_{q, r}(l, jG) \leq 0$ among all nodes of same level. Note that this invariant holds at the beginning of the phases because $\delta_{q, r}(0, j) \leq 0$ for all choices of $q < r$, $j$. Additionally, we also store the value $\delta_{q, r}(l, jG)$ for each search of $(jG, q, r)$.

Recall the following formula for $\delta_{q, r}(i, j)$:

\begin{align*}
    \delta_{q, r}(i, j) = P^\Sigma_{C, q}(i, j) + \sum_{q+1 \le x \le r} P_{C, x}^\Sigma(0, j) \\
    - \sum_{q \le x \le r-1} P_{C, x}^\Sigma(i, n) - P^\Sigma_{C, r}(i, j)
\end{align*}

To simulate one phase, we do the following. For each search $(jG, q, r)$, let's say its child represents the index $\langle p_0:p_1 \rangle, \langle p_1:p_2 \rangle, \ldots, \langle p_{H-1}:p_H \rangle$. To find the rightmost child while adhering to invariant, it suffices to compute $\delta_{q, r}(p_{i+1}, jG) - \delta_{q, r}(p_i, jG)$ for all $0 \le i \le H - 1$. Hence, for each query of $cmp(jG, q, r)$, we try to retrieve the following package of information:

\begin{itemize}
    \item For all childs $w$ of $V(jG, q, r)$ and each color $d \in [H]$, the number of elements in $arr(w, d)$. This corresponds to the term $P^\Sigma_{C, d}(p_{i+1}, n) - P^\Sigma_{C, d}(p_i, n)$.
    \item For all childs $w$ of $V(jG, q, r)$ and each color $d \in [H]$, the number of elements in $arr(w, d)$ with value at most $jG$. This corresponds to the term $P^\Sigma_{C, d}(p_{i+1}, jG) - P^\Sigma_{C, d}(p_i, jG)$.
\end{itemize}

Each package is of size $O(H^2)$, and we need $O(H^2)$ search to retrieve each package. As there are $H^2$ candidates of $q, r$, we need to process $\frac{n H^4}{G}$ total rank search queries for retrieving this information. As $\frac{n H^4}{G}$ is sublinear, all queries fit the total space, and we can apply \cref{lem:tools_search} to compute all packages in $O(1)$ rounds. From the packages, we can compute $\delta_{q, r}(p_i, j)$ for all $0 \le i \le H$. If $\delta_{q, r}(p_H, j) \le 0$, we say $cmp(jG, q, r) = p_H + 1$ and the search is done. Otherwise, we take the maximum $i$ with $\delta_{q, r}(p_i, j) \le 0$ and descend the $(i+1)$-th child from the left. Since all $\delta_{q, r}(p_i, j)$ are known, the value $\delta_{q, r}(l, j)$ can be appropriately updated. 

After $h(T)$ phases of such iteration, all searches are either finished or $V(jG, q, r)$ points to a leaf node. If a leaf node represents $p(\hat{i})$, then we can state that $cmp(jG, q, r) = i + 1$. 

Finally, we need a small technical detail to proceed to the later stage of the algorithm. For all $i$ where $opt(i, jG) \neq opt(i+1, jG)$, we compute, for each $opt(i, jG) \le k \le opt(i+1, jG)$, the value $\delta_{k, k+1}(i, jG)$. To compute this, note that the interval $\langle i, n\rangle$ can be decomposed into at most $B \times h(T)$ disjoint intervals where each interval corresponds to a node in a tree. We can find these intervals (set of nodes) by descending the tree $T$ in $h(T) = O(1)$ rounds. Then, we can use \cref{lem:tools_search} to retrieve the values from these nodes in $O(1)$ rounds. The reason for computing this information will be soon evident in \cref{sec43}. 

This concludes the descriptions of the computation of $opt(*, jG)$ for all $j$ in $O(1)$ round. 

We take a transpose of $P_C$ by taking an inverse permutation as in \cref{lem:tools_inv} and also compute all $opt(iG, *)$ for all $i$ in $O(1)$ rounds.

\subsection{Computing the remaining $opt(i, j)$ locally}
\label{sec43}

Before proceeding to this step, we present some definitions. For each $1 \le q \le H - 1$, we define a \textbf{demarcation line} $q$ as the line demarcating the area with $opt(i, j) \le q$ and $opt(i, j) > q$. Due to \cref{lem:inc1}, \ref{lem:inc2}, each demarcation line is a monotone chain connecting the lower-left and upper-right areas. We visualize the demarcation line as a polyline connecting the half-integer points $(\hat{i}, \hat{j})$. Specifically, a point $(\hat{i}, \hat{j})$ is \textbf{pierced} by demarcation line $q$ if the demarcation line passes through $(\hat{i}, \hat{j})$. This happens, if not all of $[opt(i, j), opt(i+1, j), opt(i, j+1), opt(i+1, j+1)]$ are at most $q$, or at least $q+1$. In this sense, a demarcation line $q$ connects each adjacent pierced point. A \textbf{subgrid} $(i, j)$ contains all cells in $\langle iG, (i+1)G \rangle \times \langle jG, (j+1)G \rangle$. A point $(\hat{i}, \hat{j})$ is \textbf{interesting} if $opt(i, j) = opt(i + 1, j)$, $opt(i, j) = opt(i, j+1)$, $opt(i, j) \neq opt(i+1, j+1)$. Recall that from \cref{lem:char1}, \ref{lem:char2}, \ref{lem:char3}, \ref{lem:char4}, it suffices to find and mark all interesting points.

Our first observation is that there are at most $O(\frac{nH}{G})$ subgrids of interest:

\begin{lemma}
\label{lem:subgrid_trivial}
There exists at most $O(\frac{nH}{G})$ subgrid which contains at least one interesting point.
\end{lemma}\begin{proof}
    Each demarcation line intersects at most $\frac{2n}{G}$ subgrids, as each movement either increases the column number or decreases the row number. Hence, at most $\frac{2nH}{G}$ subgrids are intersected by at least one demarcation line. If any demarcation line does not intersect a subgrid, every point in the subgrid shares the same $opt(i, j)$ value. Hence, there are no interesting points.
\end{proof}

Now, we demonstrate the algorithm for the subgrid $(i, j)$, which contains at least one interesting point. We will compute all $opt(r, c)$ in increasing order of rows and decreasing order of columns. As an invariant, we maintain the following information before computing the $opt(r, c)$ for some $iG < r \le (i+1)G, jG \le c < (j+1)G$. Let $chain(r, c)$ be the ordered sequence of cells $\{(r-1, jG), (r-1, jG+1), \ldots, (r-1, c), (r-1, c+1), (r, c+1), (r, c+2), \ldots, (r, (j+1)G), (r+1, (j+1)G), \ldots, ((i+1)G, (j+1)G)\}$. We have

\begin{itemize}
    \item The value $opt(r^\prime, c^\prime)$ for each $(r^\prime, c^\prime) \in chain(r, c)$. Note that all points in $chain(r, c)$ are in nondecreasing order of rows and columns, and the value $opt(r^\prime, c^\prime)$ are also nondecreasing.
    \item Suppose that for two adjacent cell $(r_1, c_1)$ and $(r_2, c_2)$ in $chain(r, c)$ has $opt(r_1, c_1) \neq opt(r_2, c_2)$. For all $opt(r_1, c_1) \le k < opt(r_2, c_2)$, we store $\delta_{k, k + 1}(r_1, c_1)$.
\end{itemize}

We will refer to this information as \textbf{invariant information}.

We will compute the value $opt(r, c)$. Since we have $opt(r-1, c) \le opt(r, c) \le opt(r, c+1)$, if $opt(r-1, c) = opt(r, c+1)$, we can say $opt(r, c) = opt(r-1, c)$, and it is straightforward to maintain the invariants as well. Suppose $opt(r-1, c) < opt(r, c+1)$. We need to know, for each $opt(r-1, c) \le k < opt(r, c+1)$, the value of $\delta_{k, k+1}(r, c)$. Here, since $\delta_{p, q}(r, c) = \sum_{i = p}^{q-1} \delta_{i, i+1}(r, c)$, this information is enough to determine the best subproblem. By the invariant, we either know $\delta_{k, k+1}(r-1, c)$ or $\delta_{k, k+1}(r-1, c+1)$, and we want to transform them into $\delta_{k, k+1}(r, c)$.

The proof of \cref{lem:inc1}, \ref{lem:inc2} reveals the necessary information required to obtain $\delta_{p, q}(i, j+1) - \delta_{p, q}(i, j)$ and $\delta_{p, q}(i+1, j) - \delta_{p, q}(i, j)$. From these, we can conclude that the additional information other than invariant required to acquire $\delta_{k, k+1}(r, c)$ is:

\begin{itemize}
    \item Whether the point in row $\hat{r}-1$ is of color range $[opt(r-1, c), opt(r, c+1)]$, and if it is, the color and column index of such point.
    \item Whether the point in column $\hat{c}$ is of color range $[opt(r-1, c), opt(r, c+1)]$, and if it is, the color and row index of such point.
\end{itemize}

We will refer to this information as \textbf{non-invariant information}.

Given the non-invariant information, we can determine the value $\delta_{k, k+1}(r, c)$ and $opt(r, c)$ in turn. Then we use the value $\delta_{k, k+1}(r, c)$ or $\delta_{k, k+1}(r-1 , c)$ to maintain the invariant, which is enough. Note that, when determining $opt(r, c)$, we also have the information $opt(r-1, c), opt(r-1, c+1), opt(r, c+1)$ as well, so we can also mark all interesting points as well.

This concludes the description of the algorithm for computing all interesting points. Now, let's check whether we can supply all information within our space restrictions and construct an algorithm that does it.

We will define \textbf{subgrid instance} as a set of invariant / non-variant information necessary to determine all interesting points. From \cref{lem:subgrid_trivial}, we have at most $O(\frac{nH}{G})$ subgrid instances to prepare. 

To prepare the invariant information, we can use the results of \cref{sec42}, which is of size $O(H)$ for each subgrid instance. 

One naive way to prepare the non-invariant information is to add all points in row $r \in \langle iG:(i+1)G\rangle$ and column $c \in \langle jG:(j+1)G\rangle$ to the subgrid instances, but this is of $O(G)$ size for each subgrid instances which does not sum to near-linear total size. However, we present a simple variation of this naive method that can yield a near-linear total space.

For each demarcating line $q$ and the subgrid it intersects, let $(\hat{r_1}, \hat{c_1})$ be the lower-leftmost pierced point and $(\hat{r_2}, \hat{c_2})$ be the upper-rightmost pierced point. Then, for all points in row $\langle r_1:(r_2+1) \rangle$ with color $q$ or $q+1$, we add such points to that subgrid instance. Similarly, for all points in column $\langle c_2:(c_1+1)\rangle$ with color $q$ or $q+1$, we add such points to that subgrid instance. The following lemma asserts that this procedure had added at most $O(n)$ points in total:

\begin{lemma}
    The above procedure adds one point into at most $4$ subgrid instances.
\end{lemma}
\begin{proof}
    Suppose we add a point with color $c$ into some subgrid instances. The demarcating line should either have an index $c-1$ or $c$. For a demarcating line to add a point, the row/column interval induced by two extreme pierced points should contain the point. Since each demarcating line is monotone, only two cases exist (one row, one column) in which the interval can contain a point. 
\end{proof}

We show that the subgrid instance populated by the above procedure has all non-invariant information. If $opt(r-1, c) < opt(r, c+1)$, then for all $opt(r-1, c) \le k < opt(r, c+1)$, the demarcating line $k$ pierces the point $(\hat{r} - 1, \hat{c})$. If a demarcating line $k$ pierces the point $(\hat{r} - 1, \hat{c})$, and the point in the row $\hat{r} - 1$ or column $\hat{c}$ has a color $k$ or $k+1$, then it is added to the subgrid instance. As a result, if the point in row $\hat{r} - 1$ is of color range $[opt(r-1, c), opt(r, c+1)]$, the demarcating line adds such point to the subgrid instance. The same goes for the column $\hat{c}$.

Now, we will present the actual algorithm. We first compute the size of each subgrid instance by tracing the demarcation lines in parallel. Given the size of each subgrid instance, we can sort them in the order of decreasing sizes and use greedy packing to assign the respective machines. We then retrace each demarcation line to populate the subgrid instance with non-invariant information. We also populate them with the invariant information as well. Finally, each machine solves the subgrid instance independently and reports all interesting points to complete the permutation $P_C$.

\section{Applications}\subsection{Extension to the sub-permutation case}
In this section, we prove \cref{thm:main_ext} from \cref{thm:main}. The following proof is inspired by Tiskin's proof  \cite{tiskin2013semilocal} on proving that all subunit-Monge matrices are closed under $\boxdot$ operation.

\mainExt*

\begin{proof}
For any row $i$ in $P_A$, which is zero, we can see that the row $i$ of $P_A \times P_B$ is also zero. Hence, we can delete this row and recover it from the resulting $P_C$ outputted by \cref{thm:main}. Likewise, for any column $j$ in $P_B$, which is zero, the column $j$ of $P_A \times P_B$ is also zero. Hence, using sorting in $O(1)$ rounds (\cref{lem:tools_sort}), one can assume that all rows of $P_A$ and all columns of $P_B$ are nonzero. Let $n_1, n_2, n_3$ be an integer such that $P_A$ is an $n_1 \times n_2$ matrix, and $P_B$ is an $n_2 \times n_3$ matrix. We can see that $n_1 \le n_2, n_2 \ge n_3$.

We will extend matrix $P_A$ and $P_B$ to a square $n_2 \times n_2$ matrix by appending $n_2 - n_1$ rows in the front of $P_A$, and $n_3 - n_1$ columns in the back of $P_B$. In other words, we have

\begin{align*}
P^\prime_A = \begin{bmatrix}
* \\ P_A
\end{bmatrix}, P^\prime_B = \begin{bmatrix}
P_B & *
\end{bmatrix},
\end{align*}

Then we can see, as long as $P_A^\prime$ and $P_B^\prime$ remains as a permutation matrix, the content of $*$ are irrelevent, and we can recover $P_C$ by the following:

\begin{align*}
\begin{bmatrix}
* & * \\
P_C & *
\end{bmatrix} = \begin{bmatrix}
* \\ P_A
\end{bmatrix} \times \begin{bmatrix}
P_B & *
\end{bmatrix}
\end{align*}

We have to add rows appropriately to make $P_A$ a permutation matrix. For this end, do the following: Initialize a length $n_2$ arrays $A$ filled with $1$ and mark all nonzero columns with $0$. Then, we shift the existing columns by $n_2 - n_1$, and for each element with $A[i] = 1$, we add $\sum_{j \le i} A[j]$-th row where the column $i$ is marked. As we can compute the prefix sum in $O(1)$ round \cite{10.1007/978-3-642-25591-5_39}, this procedure takes $O(1)$ round. The case of $P_B$ is just a transpose of the above procedure. 

As a result, we can convert the $P_A$ and $P_B$ into a permutation matrix $P^\prime_A, P^\prime_B$ in $O(1)$ rounds. By \cref{thm:main}, their multiplication takes $O(1)$ rounds. After the multiplication, we can simply ignore the upper $n_2 - n_1$ and right $n_3 - n_1$ columns to obtain $P_C$ in $O(1)$ rounds.
\end{proof}

\subsection{Proof of \cref{thm:main_lis}}
In this section, we prove \cref{thm:main_lis} using \cref{thm:main_ext}. Note that this is a rather typical step for the LIS algorithm that uses the unit-Monge matrix multiplication, and Section 4 of \cite{chs23} already well explains how this is done. In this paper, we delegate most of the explanation to \cite{chs23} and will only visit MPC-specific details.

\begin{proof}
    We proceed identically as the proof of Theorem 1.2 in \cite{chs23}, with only the following difference: After the splitting of $A$ into two halves $A = A_{lo} \circ A_{hi}$, we can relabel $A_{lo}$ and $A_{hi}$ in $O(1)$ rounds, since the sorting and inverse permutation can be computed in $O(1)$ rounds due to \cref{lem:tools_inv}, \ref{lem:tools_sort}.
\end{proof} 

\section{Discussion}
This paper presents an $O(1)$-round fully-scalable MPC algorithm for computing the subunit-Monge matrix multiplication. This result implies a $O(\log n)$-round fully-scalable MPC algorithm for computing the exact length of the longest increasing subsequences. Our results are deterministic, and improves all currently known algorithms of \cite{ims17, chs23}, only except the $O(1)$-round algorithm of \cite{ims17} which is approximate and is not fully-scalable.

One important question to address in future work is whether we can apply a similar approach to obtain an $O(1)$-round algorithm for the longest increasing subsequences. Since the crux of our work is to devise an algorithm for merging $poly(n)$ subproblems instead of two subproblems, it is natural to wonder if we can multiply $poly(n)$ matrices at once instead of two. Indeed, such an algorithm will likely yield an $O(1)$-round algorithm. However, our method of multiplying two matrices (at least so far as we understand the problem) requires divide-and-conquer with heavy technical details. We do not see or expect a solution that would make this work for a polynomial number of matrices. Rather, we believe there should be an $\Omega(\log n)$ conditional hardness result for the exact LIS problem, of which we currently do not have proof. Finding a very efficient $O(1)$-round algorithm or proving that such does not exist under popular conjectures will be an interesting challenge.
\bibliographystyle{alpha}
\bibliography{library}

\end{document}